\documentclass[11pt,leqno]{amsart}

\usepackage{amssymb,amsmath,amsthm,amsfonts}
\usepackage {latexsym}

\newcommand{\Hmm}[1]{\leavevmode{\marginpar{\tiny%
$\hbox to 0mm{\hspace*{-0.5mm}$\leftarrow$\hss}%
\vcenter{\vrule depth 0.1mm height 0.1mm width \the\marginparwidth}%
\hbox to 0mm{\hss$\rightarrow$\hspace*{-0.5mm}}$\\\relax\raggedright #1}}}
\newcommand{\nc}{\newcommand}
\nc{\les}{\lesssim}
\nc{\nit}{\noindent}
\nc{\nn}{\nonumber}
\nc{\D}{\partial}
\nc{\diff}[2]{\frac{d #1}{d #2}}
\nc{\diffn}[3]{\frac{d^{#3} #1}{d {#2}^{#3}}}
\nc{\pdiff}[2]{\frac{\partial #1}{\partial #2}}
\nc{\pdiffn}[3]{\frac{\partial^{#3} #1}{\partial{#2}^{#3}}}
\nc{\abs}[1] {\lvert #1 \rvert}
\nc{\cAc}{{\cal A}_c}
\nc{\cE}{{\cal E}}
\nc{\cF}{{\cal F}}
\nc{\cP}{{\cal P}}
\nc{\cV}{{\cal V}}
\nc{\cQ}{{\cal Q}}
\nc{\cGin}{{\cal G}_{\rm in}}
\nc{\cGout}{{\cal G}_{\rm out}}
\nc{\cO}{{\cal O}}
\nc{\Lav}{{\cal L}_{\rm av}}
\nc{\cL}{{\cal L}}
\nc{\cB}{{\cal B}}
\nc{\cZ}{{\cal Z}}
\nc{\cR}{{\cal R}}
\nc{\cT}{{\cal T}}
\nc{\cY}{{\cal Y}}
\nc{\cX}{{\cal X}}
\nc{\cXT}{{{\cal X}(T)}}
\nc{\cBT}{{{\cal B}(T)}}
\nc{\vD}{{\vec \mathcal{D}}}
\nc{\efield}{\mathcal{E}}
\nc{\vE}{{\vec \efield}}
\nc{\vB}{{\vec \mathcal{B}}}
\nc{\vH}{{\vec \mathcal{H}}}
\nc{\ty}{{\tilde y}}
\nc{\tu}{{\tilde u}}
\nc{\tV}{{\tilde V}}
\nc{\Pc}{{\bf P_c}}
\nc{\bx}{{\bf x}}
\nc{\bX}{{\bf X}}
\nc{\bXYZ}{{\bf XYZ}}
\nc{\bY}{{\bf Y}}
\nc{\bF}{{\bf F}}
\nc{\bS}{{\bf S}}
\nc{\dV}{{\delta V}}
\nc{\dE}{{\delta E}}
\nc{\TT}{{\Theta}}
\nc{\dPsi}{{\delta\Psi}}
\nc{\order}{{\cal O}}
\nc{\Rout}{R_{\rm out}}
\nc{\eplus}{e_+}
\nc{\eminus}{e_-}
\nc{\epm}{e_\pm}
\nc{\eps}{\varepsilon}
\nc{\vnabla}{{\vec\nabla}}
\nc{\G}{\Gamma}
\nc{\w}{\omega}
\nc{\mh}{h}
\nc{\mg}{g}
\nc{\vphi}{\varphi}
\nc{\tlambda}{\tilde\lambda}
\nc{\be}{\begin{equation}}
\nc{\ee}{\end{equation}}
\nc{\ba}{\begin{eqnarray}}
\nc{\ea}{\end{eqnarray}}

\nc{\g}{\gamma}
\nc{\ol}{\overline}

\newtheorem{obs}{Observation}
\newtheorem{theo}{Theorem}[section]

\newtheorem{prop}{Proposition}[section]
\newtheorem{lem}{Lemma}[section]

\newtheorem{rmk}{Remark}[section]

\nc{\pT}{\partial_T}
\nc{\pz}{\partial_z}
\nc{\pt}{\partial_t}
\nc{\la}{\langle}
\nc{\ra}{\rangle}
\nc{\infint}{\int_{-\infty}^{\infty}}
\nc{\halfwidth}{6.5cm}
\nc{\figwidth}{10cm}

\nc{\nlayers}{L} \nc{\nsectors}{M} 
\nc{\indicator}{\mathbf{1}} 
\nc{\Rhole}{R_{\rm hole}} 
\nc{\Rring}{R_{\rm ring}} 
\nc{\neff}{n_{\rm eff}} 
\nc{\Frem}{F_{\rm rem}} 
\nc{\Real}{\mathbb R} 
\nc{\Z}{\mathbb Z} 
\nc{\DD}{\Delta}
\nc{\cD}{\mathcal D} 
\nc{\lnorm}{\left\|} 
\nc{\rnorm}{\right\|} 
\nc{\rnormp}{\right\|_{\ell^{p,\eps}}} 
\nc{\rar}{\rightarrow}
\date{\today}

\begin{document}

\title[Quasi-linear dynamics in NLS]{Quasi-linear dynamics in nonlinear Schr\" odinger equation with periodic
boundary conditions}

\author{M. Burak Erdo\smash{\u{g}}an and Vadim Zharnitsky}
\thanks{The authors were partially supported by NSF grants DMS-0505216 (V.~Z.) and DMS-0600101 (B.~E.)}
\address{Department of Mathematics \\
University of Illinois \\
Urbana, IL 61801, U.S.A.} \email{berdogan@math.uiuc.edu \\ vz@math.uiuc.edu}

\maketitle

\begin{abstract}
It is shown that a large subset of initial data with finite energy 
($L^2$ norm) evolves nearly linearly in nonlinear Schr\" odinger equation 
with periodic boundary conditions. These new solutions are not perturbations of the 
known ones such as solitons, semiclassical or weakly linear solutions.
\end{abstract}

\section{Introduction} 
The nonlinear Schr\" odinger (NLS) equation 
\ba
iq_t+\Delta q+|q|^2 q =0,
\ea
where $q: \Real_t \times M_x \rightarrow {\mathbb C}$,
frequently appears as the leading approximation of  
the envelope dynamics of a quasi-monochromatic plane wave propagating in a weakly nonlinear 
dispersive medium. It arises in a number of physical models  in the description of nonlinear waves 
such as the propagation of  a laser beam in a medium whose index of reflection is sensitive to 
the wave amplitude.  

NLS has been considered on various domains 
such as  $M=\Real^n, {\mathbb T^n}$, with periodic or Dirichlet boundary conditions. One dimensional cubic
NLS is integrable \cite{ZakSha} and the explicit (or approximately explicit) solutions can be obtained 
as solitons, cnoidal waves, and their perturbations. There have been also many interesting results on the long 
time asymptotics of solutions of integrable NLS in the limit of small dispersion, see {\em e.g.} 
the recent monograph \cite{KamMcL}, \cite{DeiZho,TovVen,BioKod} and references therein. 

Recent results in optical communication literature (see, {\em e.g.} \cite{BerDav,Ess, MamMam}, and the 
appendix) 
suggest that for some initial data (highly localized pulses) the evolution is nearly linear.
 Based on these studies, 
we introduce a large class of solutions, which we call quasi-linear,  for one dimensional cubic NLS with periodic 
boundary conditions. These solutions can be characterized by the magnitude of Fourier coefficients 
of the initial data. We prove that these solutions evolve nearly 
linearly using  a normal form reduction and  
 estimates on Fourier sums. Although we do not explicitly use integrability, we do rely on  the integrability of 
the quartic normal form which is partially responsible 
for quasi-linear behavior. Therefore, similar results can be obtained for some 
nonlinear PDEs, such as $iq_t+ q_{xxxx}+|q|^2 q =0$, for which there are no integrability results.
We do not study long time asymptotics but rather the finite time dynamics in the limit of 
spectral broadening of initial data. This broadening forces $\|q(x,0) \|_{H^s}$ to grow to infinity, making the 
analysis rather nontrivial even for the finite time interval.  
While, we consider the focusing case, our result holds for defocusing case as well. 
The reader will be able to see that our proof can be immediately adapted for the defocusing case, since  
nowhere our arguments rely on the nonlinearity sign.  
 
In many engineering and physics applications, nonlinearity is unavoidable while modeling and optimizing 
a linear behavior is much easier than a nonlinear one.
Therefore, it is an important question whether a nonlinear system can be made to behave linearly. 
In applied mathematics and physics literature, such a behavior has been observed in  
{\em e.g.} \cite{Abl, Ess, GLC, Mik, ManZak}. We believe that our result gives a systematic 
way to analyze this behavior in nonlinear systems when the energy is distributed over many Fourier 
harmonics.

\section{Main Results}
We consider the nonlinear Schr\"odinger equation with periodic 
boundary conditions,
\[
iq_t + q_{xx}+ 2|q|^2 q = 0, 
\]
with initial data in $q(0) \in L^2(-\pi,\pi)$.
In \cite{Bou1}, Bourgain proved the $L^2$ global well-posedness of this equation.
The numerical simulations of quasi-linear regime for light wave communication systems  
suggest that the following statement should hold (see, {\em e.g.}, \cite{Ess, Mik}) 
\begin{obs}
Assume that initial data is a localized Gaussian
\[
q(x,0) = \frac{1}{\sqrt{\eps}} \,\, e^{-\frac{x^2} {\eps^2}}h(x),
\]
 where $h(x)$ is a smooth cutoff near $x=\pm \pi/2$. Then the initial data evolves quasi-linearly,
\ba\label{eq:obser}
\|q(x,t)-e^{it(\Delta +4 P)}q(x,0)\|_2 \rightarrow 0,
\ea 
as $\eps \rightarrow 0$ and for $t \leq T$, where $T$ is a fixed positive number, 
and $P=\|q(\cdot,0)\|_2^2/2\pi $.
\end{obs} 

We will prove \eqref{eq:obser} for a large class of initial data (including the ones above) characterized by the magnitude of Fourier coefficients.
We will use Fourier transform in the form
\begin{align*}
q(x,t)&= \sum_{n \in \Z} u(n,t) e^{i n x}\\
u(m,t)&= \frac{1}{2\pi} \int_{-\pi}^{\pi} q(x,t) e^{-i m x} dx,
\end{align*}
so that the NLS equation takes the form
\ba
i \frac{d u(m)}{d t} -m^2 u(m) + 2 \sum_{m_1+m_2-m_3=m} u(m_1) u(m_2)\bar u(m_3)=0.
\ea

Our main result is the following theorem. 
\begin{theo}\label{thm:main}
 Let $P>0$ and $C>0$ be fixed. Assume that the Fourier sequence of the initial data 
$u(n,0)=\widehat{q(\cdot,0)}(n)$ satisfies
$$ \|u(\cdot,0)\|_{\ell^{\infty}}\leq C \eps^{\frac12},\,\,\,\,\,\,\,\,\,\|u(\cdot,0)\|_{\ell^{1}}\leq C \eps^{-\frac12},$$
 for sufficiently small $\eps\in (0,1)$. Then, for each $t>0$, 
\begin{equation}\label{eq:mainclaim}
\|q(\cdot,t)-e^{it(\Delta +4 P)}q(\cdot,0)\|_{L^2} \les \la t \ra \,\eps^{1-},
\end{equation}
where $P=\|q(\cdot,0)\|_2^2/2\pi$, $\la t \ra=\sqrt{1+t^2}$ and the implicit constant depends only on $C$.

\end{theo}
\begin{rmk}
The initial data in the observation above satisfies the hypothesis of the theorem. In fact, if $f$ is an $H^s$ function for some $s>1$   with 
compact support on $(-\pi,\pi)$, then 
\[
f_\eps(x)= \frac{1}{\sqrt{\eps}}f(x/\eps)
\]
satisfies the hypothesis of the theorem.
\end{rmk}

By continuous dependence on initial data in $L^2$, it suffices to prove \eqref{eq:mainclaim} for any $\delta>0$ and for any initial data in the following subset of $L^2$: 
\[
B^{\delta}_{\eps, C}=\{f \in L^2: \|\hat f\|_{\ell^{p,\delta}}:=\left[\sum_{n=-\infty}^\infty |\hat f(n)|^p e^{\delta
|n|p}\right]^{1/p}\leq C\eps^{\frac12-\frac1p},  p\in [1, \infty]\}.
\]
Since $B^{\delta}_{\eps, C}\subset H^1$, 
we can introduce the Hamiltonian 
\cite{Kuk} 
\[
H(u)= i\sum_{n} n^2 |u(n)|^2 - i\sum_{l(n)=0}
u(n_1)u(n_2)\bar u(n_3)\bar u(n_4),
\]
with conjugated variables $\{u(n), \bar u(n)\}_{n \in \Z}$,
where $l(n)=n_1+n_2-n_3-n_4$. The Hamiltonian flow is then
given by
\[
\dot u(n) = \frac{\partial H}{\partial \bar u(n)}\cdot
\]
 
Theorem~\ref{thm:main} follows from the following by continuous dependence on initial data in $L^2$.
\begin{theo}\label{thm:main1} Let $P>0$ and  $C>0$  be fixed.
 Assume that $\|q(0)\|_2^2=2 \pi P$, and $q(\cdot,0) \in B^\delta_{\eps,C}$  for some $\delta>0$, and for sufficiently small $\eps\in(0,1)$.
Then, for each $t>0$,
\be 
\|q(\cdot,t)-e^{it(\Delta + 4P)}q(\cdot,0)\|_2 \les \la t\ra \, \eps^{1-}, \label{eq:near_lin}
\ee 
 where the implicit constant depends only on $C$.
\end{theo}

The proof of Theorem~\ref{thm:main1} is based on the normal form transformations, see, e.g., \cite{Kuk}, \cite{KukPos} and \cite{Bou2}. 
In Section~\ref{sec:norform}, we introduce a canonical transformation $u=u(v)$ in the Fourier space which brings the equation into the form\footnote{Similar quasi-linear behavior can be obtained for the nonintegrable NLS   $iq_t+ q_{xxxx}+|q|^2 q =0$ with the leading behavior given by $\dot v(n) = i(n^4 + 4 P) v(n).$},
see \eqref{eq:trans_evol1} and \eqref{eq:remain} below, 
\begin{equation}\label{eq:vevol}
\dot v(n) = i(n^2 + 4 P) v(n)+ E(v)(n).
\end{equation}
We prove that the transformation $u=u(v)$ is near-identical in the following sense.
\begin{prop}\label{prop:nformest}
If $u\in B^\delta_{\eps,C}$ or $v\in B^\delta_{\eps,C}$, then
\[
\|u\|_{\ell^2} = \|v\|_{\ell^2}, \text{ and }\,\,\,\,\,\,\,\, \|u-v\|_{\ell^{p,\delta}} \les  \eps^{\frac{3}{2}-\frac{1}{p}-}
\]
for $1\leq p \leq \infty$, where the implicit constant depends on $C$ and $p$. \\
In particular, if $\eps$ is sufficiently small, then $u\in B^\delta_{\eps,C}$ implies $v\in B^\delta_{\eps,2C}$ and vice versa.
\end{prop}
\noindent
Then, we estimate the error term $E(v)$ as follows 
\begin{prop}\label{prop:error_estimate}
If $v\in B^\delta_{\eps,C}$, then the error term $E(v)$ in the transformed equation \eqref{eq:vevol} satisfies 
\[
\|E(v)\|_{\ell^{p,\delta}} \les \eps^{\frac32-\frac1p-},
\]
for $1\leq p \leq \infty$, where the implicit constant depends on $C$ and $p$.
\end{prop}
\noindent
Propositions~\ref{prop:nformest} and \ref{prop:error_estimate} imply Theorem~\ref{thm:main1}. 
Indeed, assume that $q(\cdot,0)\in B^\delta_{\eps,C}$ for some $\delta>0,$ $C>0$, and for sufficiently small $\eps\in(0,1)$.
Multiplying \eqref{eq:vevol} with $e^{-i(n^2+4P)t}$ and integrating over $t$, we obtain
\[
v(n,t)e^{-i(n^2+4P)t}-v(n,0) = \int_0^t e^{-i(n^2+4P)\tau}E(v)d\tau.
\]
This and Propositions~\ref{prop:nformest} and \ref{prop:error_estimate} imply, for each $p\in [1,\infty]$,  that
\begin{equation}\label{eq:vevolve}
 \|v(t)-e^{iLt}v(0)\|_{\ell^{p,\delta}}=\|v(t)e^{-iLt}-v(0)\|_{\ell^{p,\delta}}  \les t \, \eps^{\frac32-\frac{1}{p}-},
\end{equation}
where $L(v)(n)=(n^2+4P)v(n)$. Finally, Proposition~\ref{prop:nformest} and \eqref{eq:vevolve} imply, for $p\in[1,\infty]$,  that  
\begin{align*}
\|u(t)-e^{iLt}u(0)\|_{\ell^{p,\delta}} & \leq  \|u(t)-v(t)\|_{\ell^{p,\delta}}  +\|v(t)-e^{iLt}v(0)\|_{\ell^{p,\delta}} \\&+\|e^{iLt}v(0)-e^{iLt}u(0)\|_{\ell^{p,\delta}} \\
& \les  \la t \ra \, \eps^{\frac32-\frac{1}{p}-}, 
\end{align*}
where the implicit constant depends on $C$.
In particular, this yields the assertion of Theorem~\ref{thm:main1} as follows
\[
 \|q(t)-e^{it(\Delta+2P)}q(0)\|_2=\|u(t)-e^{iLt}u(0)\|_{\ell^{2}} \leq  \|u(t)-e^{iLt}u(0)\|_{\ell^{2,\delta}} \les \la t \ra \, \eps^{1-}.
\]

\vspace{10mm}
\noindent
{\bf \large Notation.} \\
We will frequently use convolution with $1/|n|$, which will be denoted by 
$
\rho(n)=\frac{1}{|n|} \chi_{\Z\backslash \{0\}}(n).
$
and we will also use the notation $\la n \ra = \sqrt{1+n^2}$.
\\
We always assume  by default that the summation index avoids the terms with vanishing denominators. \\
To avoid using unimportant constants, we will use $\les$ sign: \\
$A\les B$ means there is an absolute constant $K$ such that $A\leq K B$. In some cases the constant will depend on parameters such as $p$.\\
$A\les B(\eta-)$ means that for any $\gamma>0$, $A\leq C_\gamma B(\eta-\gamma)$. \\
$A\les B(\eta+)$ is defined similarly.

\section{Normal form calculations} \label{sec:norform}
Consider the change of variables $u_n \rightarrow v_n$, generated by
the time 1 flow of a purely imaginary Hamiltonian $F$. 
Namely, solve
\[
\frac{d w}{ds} = \frac{\D F}{\D \bar w},\,\,\,\,\,w|_{s=0}=v,
\]
 thus producing a symplectic
transformation  $u = u(v):=w|_{s=1}$.
Let $X_F^s$ be the time $s$ map of the flow of $F$. Using
Taylor expansion \cite{Kuk, KukPos}, we have
\begin{align}\label{eq:taylor}
H\circ X_F^1(v)& =H(v) + \{H,F\}(v)+ \ldots 
+ \frac{1}{k!}\{\ldots\{\{H,\underbrace{F\},F\},\ldots,F}_{k }\}(v)  \\
& + \int_0^1 \frac{(1-s)^k}{k!} \{\ldots\{\{H,\underbrace{F\},F\},\ldots,F}_{k+1}\}\circ X_F^s(v) \,ds,   \nonumber
\end{align}
where
\ba
\{A,B \} =  \sum_n  \left ( \frac{\D A}{\D u(n)} \frac{\D B}{\D \bar u(n)} -
           \frac{\D A}{\D \bar u(n)} \frac{\D B}{\D  u(n)} \right )
\ea
is the Poisson bracket.

Recall that $H$ has a quadratic and a quartic part
\begin{equation}\label{eq:hdecom}
H = \Lambda_2 + H_4,
\end{equation}
where
\begin{equation} 
\Lambda_2=i\sum m^2 |u(m)|^2.
\end{equation}
We write $H_4=H_4^{\rm nr}+H_4^{\rm r}$, where the superscripts ``nr'' and ``r'' denotes the non-resonant and resonant
terms:
\begin{align*}
H_4^{\rm nr} &=  i\sum_{ l(m)=0,\,q(m)\neq 0} v(m_1)v(m_2)\bar v(m_3) \bar v(m_4)\\
H_4^{\rm r} &=i \sum_{l(m)=0,\,q(m)=0} v(m_1)v(m_2)\bar v(m_3) \bar v(m_4),
\end{align*}
where $l(m)=m_1+m_2-m_3-m_4$ and $q(m)=m_1^2+m_2^2-m_3^2-m_4^2$. As usual $H_4^{\rm r}$ is the part of the Hamiltonian that commutes with $\Lambda_2$. 
Note that we can further decompose $H_4^{\rm r}$ as 
\begin{align*}
H_4^{\rm r}& =-i\sum_{m}|v(m)|^4 + 2i\sum_{m_1,m_2} |v(m_1)|^2 |v(m_2)|^2 :=H_4^{\rm r1}+H_4^{\rm r2}.
\end{align*}

We sequentially apply two normal form transformations generated by $F_1$ and $F_2$. We choose $F_1$  so that the following cancellation property holds
\begin{equation}\label{eq:lambdaF}
\{ \Lambda_2,F_1\} = -H_4^{\rm nr}.
\end{equation}
We will prove that $F_1$ commutes with $H_4^{\rm r_2}$. Using these cancellation properties in \eqref{eq:taylor} with $k=2$, we obtain
\begin{align*}
H\circ X_{F_1}^1&=\Lambda_2 + H_4^{\rm r1} + H_4^{\rm r2}+\{H_4^{\rm  r1},F_1\}+\frac12\{H_4^{\rm nr},F_1\}+\frac12 g_{F_1}^{2}( H_4)
\\ 
&+ \int_0^1 \frac{(1-s)^2}{2} 
g_{F_1}^{3}(H)\circ X_{F_1}^s \, ds,
\end{align*}
where we used the notation 
 \[
 g_F^0(H) =H, \,\,\,\,\,\,\,\, g_F^{k+1}(H)= \{g_F^k,F \},\,\,\,\, k=0,1,2,\ldots
 \]

Now, we apply the second transformation\footnote{\label{fn:gf1h}It turns out
that the transform generated by  $F_1$ is not enough since the term $\{H_4^{\rm nr},F_1\}$ is present in the Hamiltonian.  The direct estimate of this term produces finite order nonlinear effect (see Subsection~\ref{sec:gf2h4}).} generated by    $F_2$ to eliminate the non-resonant terms in 
$\frac{1}{2} \{ H_4^{\rm nr}, F_1 \}$, i.e.,
\begin{equation}
\{ \Lambda_2,F_2\} = -\frac{1}{2}\{H_4^{\rm nr},F_1\}^{\rm nr}.
\end{equation}
We will also prove that $F_2$ commutes with $H_4^{\rm r_2}$. Using these cancellation properties as above in \eqref{eq:taylor} (with $k=1$), we obtain
$$
H\circ X_{F_1}^1\circ X_{F_2}^1=\Lambda_2 +  H_4^{\rm r2}+ R,
$$
where
\begin{align} \nonumber
R&=  H_4^{\rm r1}   +\{H_4^{\rm  r1},F_1\}+\frac12\{H_4^{\rm nr},F_1\}^{\rm r}+\{H_4^{\rm  r1},F_2\}+\{\{H_4^{\rm r1},F_1\},F_2\}+K
\\ 
&+ \frac12\{\{H_4^{\rm nr},F_1\},F_2\}+\{K,F_2\}+ \int_0^1 (1-s) 
g_{F_2}^{2}(H\circ X_{F_1}^1)\circ X_{F_2}^s \, ds, \nonumber
\end{align}
where 
$$
K=\frac12 g_{F_1}^{2}( H_4) + \int_0^1 \frac{(1-s)^2}{2} 
g_{F_1}^{3}(H)\circ X_{F_1}^s \, ds.
$$

The transformed evolution equation is given by
\begin{equation}\label{eq:trans_evol}
\dot v(n) = \frac{\D (H\circ X_{F_1}^1\circ X_{F_2}^1) }{\D \bar v}. 
\end{equation}
Note that contribution of the ``leading'' terms,  $\Lambda_2+H_4^{\rm r2}$, is given by
\[
\frac{\D }{\D \bar v(n)}\Big(i\sum m^2 |v(m)|^2 +  2i\sum_{m_1,m_2} |v(m_1)|^2 |v(m_2)|^2\Big)=i (n^2 + 4 P) v(n).
\]
Therefore, we can rewrite \eqref{eq:trans_evol} as
\begin{equation}\label{eq:trans_evol1}
\dot v(n) = i (n^2 + 4 P) v(n)+E(v)(n), 
\end{equation}
where

\begin{equation}\label{eq:remain}
E(v)(n)= \frac{\D R}{\D \bar v(n) }.
\end{equation}

\subsection{Calculation of $F_1$ and $F_2$} \label{sec:f1f2}
To obtain \eqref{eq:lambdaF}, we take $F_1$ of the form
\be
F_1 = \sum_{l(m)=0} f(m_1,m_2,m_3,m_4) v(m_1)v(m_2) \bar
v(m_3) \bar v(m_4).  \nonumber
\ee
We have
\begin{multline} \nonumber
\{\Lambda_2,F_1  \} = i \sum_m m^2 \left ( \bar v(m) \frac{\D F_1}{\D \bar v(m)}-
v(m) \frac{\D F_1}{\D  v(m)}   \right ) \\
= i\sum_{l(m)=0} (m_1^2+m_2^2-m_3^2-m_4^2)f(m_1,m_2,m_3,m_4) v(m_1)v(m_2)\bar v(m_3) \bar v(m_4). 
\end{multline} 
Therefore, we let 
\begin{equation}\label{eq:F2ndform}
F_1 \!=\! \sum_{l(m)=0}\!\!
\frac{v(m_1)v(m_2)\bar v(m_3) \bar v(m_4)}{m_1^2+m_2^2-m_3^2-m_4^2} \!=\!\sum_{l(m)=0}\!\! \frac{v(m_1)v(m_2)\bar v(m_3)\bar v(m_4)}{2(m_1-m_3)(m_2-m_3)}.
\end{equation}
Now, we calculate $F_2$.  Using the Hamiltonian structure\footnote{These identities
follow from the following easily checked ones:\\
$\Re{(H)}=0$, $\D_v H(v,\bar v)+\D_v \bar H(v,\bar v)=0$ and  
$\D_v \bar H(v,\bar v) = \overline{\D_{\bar v} H(v,\bar v)}$.}
\[\frac{\D H} {\D \bar v(n)} = -\overline{\frac{\D H}{\D v(n)}},\;\;\;\;\;\;\;
\frac{\D F_2} {\D \bar v(n)} = -\overline{\frac{\D F_2}{\D v(n)}}
\]
we obtain
\[
\{H_4^{\rm nr},F_1 \}^{\rm nr}=2 i \!\!\!\!\!\!\sum_{\stackrel{ m_4,m_5\neq m_1,\,\,m_2,m_3\neq m_6}{l(m)=0,\,\,q(m)\neq 0}}\!\!\!\frac{v(m_1)v(m_2)v(m_3)\bar v(m_4)\bar v(m_5)\bar v(m_6)}
{(m_2-m_6)(m_3-m_6)} -c.c.
\]
Therefore, a calculation similar to the one for $F_1$ yields 
\begin{equation}\label{eq:F2form}
F_2 = \sum_{\stackrel{ m_4,m_5\neq m_1,\,\,m_2,m_3\neq m_6}{l(m)=0,\,\,q(m)\neq 0}} \frac{ 
v(m_1)v(m_2)v(m_3)\bar v(m_4)\bar v(m_5)\bar v(m_6)}{q(m)(m_2-m_6)(m_6-m_3)}-c.c.
\end{equation}
Here, $l(m)=m_1+m_2+m_3-m_4-m_5-m_6$, and $q(m)=m_1^2+m_2^2+m_3^2-m_4^2-m_5^2-m_6^2$.

\subsection{Proof of Proposition~\ref{prop:nformest}}\label{sec:normform}

First we state a simple corollary of Young's inequality. Recall that $\rho(n)=1/|n|$ for $n\neq 0$ and $\rho(0)=0$.
\begin{lem}\label{lem:young}
For any $p>1$, for any choices of $\pm$ signs
$$\big\|\sum_j w(\pm n \pm j) \rho(\pm j)\big\|_{\ell^p_n} \les \|w\|_{\ell^{p-}}.$$
With some abuse of notation, we denote each sum of the above form   by  $w*\rho$.
\end{lem}

\begin{proof}
Recall that by Young's inequality, 
$\|w*\rho\|_{\ell^p}\les \|w\|_{\ell^{q}}\|\rho\|_{\ell^r}$, where $1+\frac1p=\frac1q+\frac1r$.
The lemma follows since  $\rho\in \ell^q$ for any $q>1$.
\end{proof}

\begin{proof}[Proof of Proposition~\ref{prop:nformest}]
First note that the equality of the $\ell^2$ norms follows from Hamiltonian formalism.  Indeed, it is 
straightforward to verify that $\{F,Q\} =0$ (where $Q(u)=\|u\|_2^2$), which implies $\ell^2$ norm conservation. To prove the second statement, we should estimate the time 1 map of the flow of $F_1$ and of $F_2$. We start with $F_1$.
\begin{align}
\frac{d w(n)}{d s} = \frac{\D F_1}{\D \bar w(n)}  =
\sum_{m_1+m_2-m_3-n=0}
 \frac{w(m_1)w(m_2)\bar w(m_3)}{(m_1-n)(m_2-n)}.  \label{eq:normfeq2}
\end{align}
Multiplying with $e^{\delta|n|}$, we estimate (assuming that $w\in B^\delta_{\eps,C}$)
\begin{align*}
 &\big|e^{\delta|n|}\frac{d w(n)}{ds} \big|  \leq  \\ 
 &\leq \!\!\!\!  \sum_{m_1+m_2-m_3-n=0}
\frac{e^{-\delta(|m_1|+|m_2|+|m_3|-|n|)}}{|m_1-n||m_2-n|}
|w(m_1)e^{\delta|m_1|}w(m_2)e^{\delta|m_2|}  w(m_3)e^{\delta|m_3|}|
\nonumber\\
&\leq  \|w\|_{\ell^{\infty,\delta}} \sum_{m_1,m_2}\frac{|w(m_1)|e^{\delta|m_1|}
|w(m_2)|e^{\delta |m_2|}}{|m_1-n| |m_2-n|} \leq  \|w\|_{\ell^{\infty,\delta}}
[|w| e^{\delta |\cdot|}*\rho]^2(n).
\end{align*}
In the second line, we used the fact that  $|m_1|+|m_2|+|m_3|-|n| \geq 0$. 
Therefore, by Lemma~\ref{lem:young}, we obtain
\[
\big \|\frac{d w}{ds}  \big\|_{\ell^{\infty,\delta}} \leq \|w\|_{\ell^{\infty,\delta}}
\||w| \, e^{\delta|\cdot|}*\rho\|_{\ell^{\infty}}^2 \leq \|w\|_{\ell^{\infty,\delta}}
\|w\|_{\ell^{q,\delta}}^2
\]
 for any $1\leq q<\infty$.
Similarly, using Lemma~\ref{lem:young}, we obtain
\[
\big\|\frac{d w}{ds} \big\|_{\ell^{1,\delta}} \leq \|w\|_{\ell^{\infty,\delta}}
\| |w| \, e^{\delta |\cdot|} * \rho\|_{2}^2 \les \|w\|_{\ell^{\infty,\delta}}
\|w\|_{\ell^{2-,\delta}}^2.
\]
The last two inequalities imply that if $w(0) \in B^\delta_{\eps,C}$ (or $w(1) \in B^\delta_{\eps,C}$) then 
\[
\|w(s)-w(0)\|_{\ell^{\infty,\delta}} \les \eps^{\frac{3}{2} - }, \;\;\;\;\;
\|w(s)-w(0)\|_{\ell^{1,\delta}} \les \eps^{\frac12-}.
\]
This completes the proof for  $F_1$. In the proof for $F_2$, we omit some of the details, in particular the multiplication with $e^{\delta|n|}$ argument above, since it works exactly in the same way. To estimate the $\ell^p$ norm of the right hand side of 
$$
\frac{d w(n)}{d s} = \frac{\D F_2}{\D \bar w(n)},
$$
we use duality:
\begin{equation} \label{eq:dual}
\Big\|\frac{\D F_2}{\D \bar w(n)}\Big\|_{\ell^p}=\sup_{\|h\|_{\ell^{p'}}=1}\Big|\sum h(n)  \frac{\D F_2}{\D \bar w(n)}\Big|.
\end{equation}
Note that the right hand side of \eqref{eq:dual} can be estimated by the sum of six terms of the form 
\begin{align} \label{eq:6terms}
\tilde F_2(w_1,\ldots, w_6):=\sum_{\stackrel{ m_4,m_5\neq m_1,\,\,m_2,m_3\neq m_6}{l(m)=0,\,\,q(m)\neq 0}} \frac{ 
w_1(m_1)\cdots w_6(m_6)}{|q(m)||m_2-m_6||m_6-m_3|},
\end{align} 
where in the $j$th term $w_j=|h|$ and the others are $|v|$.
The required estimates for these terms follow by applying Lemma~\ref{lem:F2} below with arbitrarily small $\eta$ and with $i=j$ if $p'=1$ and with $k=j$ if $p'=\infty$.
\end{proof}
\begin{lem} \label{lem:F2} 
For any $\eta>0$ and for any distinct $i,k \in \{1,2,3,4,5,6\}$, there is a permutation $(i_1,i_2,i_3,i_4)$ of the remaining indices  such that
$$
\tilde F_2(w_1,\ldots,w_6)\lesssim \|w_i\|_{\ell^1} \|w_k\|_{\ell^{\infty}} \|w_{i_1}\|_{\ell^1}\prod_{l=2}^4\|w_{i_l}\|_{\ell^{\infty}}^{\frac{1}{1+\eta}}\|w_{i_l}\|_{\ell^1}^{\frac{\eta}{1+\eta}}.
$$
\end{lem}
\begin{proof} Fix $\eta>0$, $i$, and  $k$. By Holder's inequality we have 
\begin{align}\label{eq:F2tilde}
\tilde F_2  &\leq \big[ \sum_{\stackrel{ m_4,m_5\neq m_1,\,\,m_2,m_3\neq m_6}{l(m)=0,\,\,q(m)\neq 0}} \frac{ 
w_1(m_1)\cdots w_6(m_6)}{|q(m)|^{1+\eta}|m_2-m_6|^{1+\eta}|m_6-m_3|^{1+\eta}}\Big]^\frac{1}{1+\eta}\\
&\times \big[ \sum_{l(m)=0}   
w_1(m_1)\cdots w_6(m_6) \Big]^\frac{\eta}{1+\eta}. \nonumber
\end{align}
The second line is bounded by 
$$\|w_k\|_{\ell^\infty}^\frac{\eta}{1+\eta} \prod_{l=1, l\neq k}^6\|w_l\|_{\ell^{1}}^\frac{\eta}{1+\eta}.$$
The required estimate for the sum in the first line follows from the following claim:
For any permutation $(j_1,j_2,j_3)$ of $\{1,4,5\}$, and for any permutation $(n_1,n_2,n_3)$ of $\{2,3,6\}$, we have
\begin{multline}\label{eq:F2sum1}
 \sum_{\stackrel{ m_4,m_5\neq m_1,\,\,m_2,m_3\neq m_6}{l(m)=0,\,\,q(m)\neq 0}} \frac{ 
w_1(m_1)\cdots w_6(m_6)}{|q(m)|^{1+\eta}|m_2-m_6|^{1+\eta}|m_6-m_3|^{1+\eta}}
\les \\
 \les \|w_{j_1}\|_{\ell^\infty}\|w_{j_2}\|_{\ell^\infty}\|w_{j_3}\|_{\ell^1}\|w_{n_1}\|_{\ell^\infty}\|w_{n_2}\|_{\ell^\infty}
\|w_{n_3}\|_{\ell^1}.
\end{multline}
To prove this inequality, replace $m_{j_1}$ in the sum with a linear combination of other indices using the identity $l(m)=0$. 
We claim that 
$$\||q(m)|^{-1-\eta}\|_{\ell^{1}_{m_{j_2}}}\les 1,$$
where the implicit constant is independent of the remaining indices. Indeed, it suffices to consider the cases $j_1=1$, $j_2=4$ and $j_1=4$, $j_2=5$ since $m_4$ and $m_5$ enter symmetrically. In the former case
\begin{align*}
q(m)&=(m_4+m_5+m_6-m_2-m_3)^2+m_2^2+m_3^2-m_4^2-m_5^2-m_6^2 \\& =C_1 m_4+C_2,
\end{align*}
where the integers $C_1, C_2$ depend on $m_2,m_3,m_5,m_6$. Moreover, $C_1\neq 0$ since $m_1\neq m_4$. Therefore,
$$
\sup_{m_2,m_3,m_5,m_6} \||q(m)|^{-1-\eta}\|_{\ell^{1}_{m_4}}\lesssim 1.
$$
In the latter case
\begin{align*}
q(m)=C_1+C_2 m_5 - 2 m_5^2,
\end{align*}
where the integers $C_1, C_2$ depend on $m_1,m_2,m_3,m_6$. Since for any integers $n, C_1, C_2$, the equation $n= C_1+C_2 m_5 - 2 m_5^2$ has at most two solutions, we have 
$$
\sup_{m_1,m_2,m_3,m_6} \||q(m)|^{-1-\eta}\|_{\ell^{1}_{m_5}}\lesssim 1.
$$
Using this claim, we obtain
\begin{align}\nonumber
\eqref{eq:F2sum1} & \les \|w_{j_1}\|_{\ell^{\infty}}\|w_{j_2}\|_{\ell^{\infty}}\|w_{j_3}\|_{\ell^1}
\sum
\frac{w_2(m_2)w_3(m_3) w_6(m_6)}{|m_2-m_6|^{1+\eta}|m_6-m_3|^{1+\eta}}\\ \nonumber
&\leq \|w_{j_1}\|_{\ell^{\infty}}\|w_{j_2}\|_{\ell^{\infty}}\|w_{j_3}\|_{\ell^1}
\|w_{n_1}\|_{\ell^{\infty}}\|w_{n_2}\|_{\ell^{\infty}}
\sum
\frac{w_{n_3}(m_{n_3})}{|m_2-m_6|^{1+\eta}|m_6-m_3|^{1+\eta}}\\
  \nonumber
&\les \|w_{j_1}\|_{\ell^{\infty}}\|w_{j_2}\|_{\ell^{\infty}}\|w_{j_3}\|_{\ell^1}
\|w_{n_1}\|_{\ell^{\infty}}\|w_{n_2}\|_{\ell^{\infty}}
 \|w_{n_3}\|_{\ell^1}.
\end{align}
 
\end{proof}

\subsection{Cancellation property of $H_4^{\rm r2}$}

We claim that
$
\{H_4^{\rm r2},F_j  \} = 0, j=1,2.
$
Indeed, by (\ref{eq:F2ndform}) and \eqref{eq:F2form}, both $F_1$ and $F_2$ have the phase invariant property
\[
F_j(v) = F_j(ve^{i\phi}),
\]
but the evolution induced by $H_4^{\rm r2}$ is just uniform phase rotation,
\[
v(n,t) = e^{i 2 P t} v(n,0).
\]
Thus,
\[
\{ H_4^{\rm r2},F_j \}:= \frac{d}{dt}  F_j(X_{H_4^{\rm r2}}^{t=0})=0,\,\,\,\,\,j=1,2.
\]

\section{Proof of Proposition~\ref{prop:error_estimate}}
Assuming that $v\in B^\delta_{\eps, C}$, we should prove that the $\ell^{p,\delta}$ norm of each of the summands in \eqref{eq:remain} is $\les \eps^{3/2-1/p-}$ for $p=1$ and $p=\infty$. To simplify the exposition, we will do this only in the case $\delta=0$. The proof for the case $\delta>0$ is similar by using the simple multiplication by $e^{\delta|\cdot|}$ argument we used in the proof of Proposition~\ref{prop:nformest}.

Note that it suffices to consider the $\D_{\bar v(k)}$ derivatives of the following terms
\begin{align*}
 H_4^{\rm r1},\,\,\, \{H_4^{\rm  r1},F_1\}, \,\,\,\{H_4^{\rm  r1},F_2\}, \,\,\,\{H_4^{\rm nr},F_1\}^{\rm r},\,\,\,
 g_{F_2}^bg_{F_1}^a(H_4), \,\, \,a+b\geq 2,
\end{align*}
and the terms involving integrals. 

We define
$$
f_1(v_1,v_2,v_3)(k):=\sum_{m_1,m_2\neq k}
 \frac{v_1(m_1)v_2(m_2) \index{\footnote{}}v_3(m_1+m_2-k)}{(m_1-k)(m_2-k)}
$$
so that $f_1(v,v,\bar v)(k)= \D_{\bar v(k)} F_1$.
Similarly we define $f_2(v_1,v_2,v_3,v_4,v_5)(k)$ so that $f_2(v,v,v,\bar v,\bar v)(k)= \D_{\bar v(k)} F_2$.
The following Lemma will be used repeatedly: 
\begin{lem} \label{lem:f1f2} 
I) For any $q\in[1,\infty]$ and any permutation $(i_1,i_2,i_3)$ of $(1,2,3)$, we have
$$\|f_1(v_1,v_2,v_3)\|_{\ell^q}\lesssim \|v_{i_1}\|_{\ell^q}\|v_{i_2}\|_{\ell^{\infty-}}\|v_{i_3}\|_{\ell^{\infty-}}. 
$$
II) For any $q \in [1,\infty]$, for any $\eta>0$, and for any $i\in\{1,2,3,4,5\}$ there is a permutation $(i_1,i_2,i_3,i_4)$ of the set $\{1,2,3,4,5\}\backslash \{i\}$ such that
$$
\|f_2(v_1,v_2,v_3,v_4,v_5)\|_{\ell^q}\lesssim \|v_i\|_{\ell^q} \|v_{i_1}\|_{\ell^1}
\prod_{l=2}^4\|v_{i_l}\|_{\ell^{\infty}}^{\frac{1}{1+\eta}}\|v_{i_l}\|_{\ell^1}^{\frac{\eta}{1+\eta}}.
$$
\end{lem}
\begin{proof}
Part I can easily be verified following the proof of Proposition~\ref{prop:nformest} with $\delta=0$. Part II follows from Lemma~\ref{lem:F2} and interpolation.
\end{proof}

\subsection{Estimate of $\D_{\bar v(k)}H_4^{\rm r1}$}  
Recall that
\[
H_4^{\rm r1} = i\sum_m |v(m)|^4,
\]
and hence
\[
\frac{\D H_4^{\rm r1}}{\D \bar v(k)} = 2i|v(k)|^2 v(k).
\]
We estimate the contribution of this term as 
\[
\Big\| \frac{\D H_4^{\rm r1}}{\D \bar v(\cdot)}\Big\|_{\ell^\infty} \les \|v^3\|_{\ell^{\infty}} \les \eps^{3/2},
\]
and
\[
\Big\| \frac{\D H_4^{\rm r1}}{\D \bar v(\cdot)}\Big\|_{\ell^1} \les\|v^3\|_{\ell^1} = \|v\|^3_{\ell^3} \les 
\eps^{\frac{3}{2}-1}.
\]
 
\subsection{Estimates for $\D_{\bar v(k)}\{H_4^{\rm r1},F_1\}$ and  $\D_{\bar v(k)}\{H_4^{\rm r1},F_2\}$.}
Let 
$$\tilde H_4^{\rm r1}(v_1,v_2,v_3,v_4):=\sum_n v_1(n)v_2(n)v_3(n)v_4(n).$$
We use duality as in \eqref{eq:dual}. Note that $\sum_k |\D_{\bar v(k)}\{H_4^{\rm r1},F_1\}| |h(k)|$
is bounded by  the sum of the following two terms  
$$\tilde H_4^{\rm r1}(|f_1(|v|,|v|,|v|)|,|h|,|v|,|v|),\,\,\,\,\,\,\tilde H_4^{\rm r1}(|f_1(|h|,|v|,|v|)|,|v|,|v|,|v|),$$ 
 and similar terms obtained by permuting the arguments. The following estimates (with $p=1$ and $p=\infty$), which follow from the definition of $\tilde H_4^{\rm r1}$ and Lemma~\ref{lem:f1f2}, completes the analysis of $\D_{\bar v(k)}\{H_4^{\rm r1},F_1\}$:
\begin{align*}
\tilde H_4^{\rm r1}(|f_1(|v|,|v|,|v|)|,|h|,|v|,|v|) &\les \|h\|_{\ell^{p'}}\|v\|_{\ell^p}\|v\|_{\ell^\infty}\|f_1(|v|,|v|,|v|)\|_{\ell^\infty}\\
&\les \|h\|_{\ell^{p'}}\|v\|_{\ell^p}\|v\|_{\ell^\infty}^2 \|v\|_{\ell^{\infty-}}^2  \les \eps^{\frac52-\frac1p-}.
\end{align*}
\begin{align*}
\tilde H_4^{\rm r1}(|f_1(|h|,|v|,|v|)|,|v|,|v|,|v|) &\les \|f_1(|h|,|v|,|v|)\|_{\ell^{p'}}\|v\|_{\ell^p}\|v\|_{\ell^\infty}^2 \\
&\les \|h\|_{\ell^{p'}}\|v\|_{\ell^p}\|v\|_{\ell^\infty}^2 \|v\|_{\ell^{\infty-}}^2  \les \eps^{\frac52-\frac1p-}.
\end{align*}
We estimate $\D_{\bar v(k)}\{H_4^{\rm r1},F_2\}$ similarly. The estimates below imply the required bound
\begin{align*}
\tilde H_4^{\rm r1}(|f_2(|v|,\ldots,|v|)|,|h|,|v|,|v|) &\les \|h\|_{\ell^{p'}}\|v\|_{\ell^p}\|v\|_{\ell^\infty}\|f_2(|v|,\ldots,|v|)\|_{\ell^\infty}\\
& \les \|h\|_{\ell^{p'}}\|v\|_{\ell^p}\|v\|_{\ell^\infty}^2  
  \|v\|_{\ell^1}
 \|v\|_{\ell^{\infty}}^{\frac{3}{1+\eta}}\|v\|_{\ell^1}^{\frac{3\eta}{1+\eta}} 
 \\& \les \eps^{\frac52-\frac1p-}.
\end{align*}
\begin{align*}
\tilde H_4^{\rm r1}(|f_2(|h|,|v|,|v|,|v|,|v|)|,|v|,|v|,|v|) &\les \|f_2(|h|,\ldots,|v|)\|_{\ell^{p'}}\|v\|_{\ell^p}\|v\|_{\ell^\infty}^2 \\
&\les \|h\|_{\ell^{p'}}\|v\|_{\ell^1}\|v\|_{\ell^p}\|v\|_{\ell^\infty}^2 \|v\|_{\ell^{\infty}}^{\frac{3}{1+\eta}}\|v\|_{\ell^1}^{\frac{3\eta}{1+\eta}}\\
&   \les \eps^{\frac52-\frac1p-}.
\end{align*}
In both estimates, the last inequality is obtained by  taking $\eta$ sufficiently small.

\subsection{Estimate of $\D_{\bar v(k)} \{H_4^{\rm nr},F_1\}^{\rm r}$.}
Based on the calculations in Section~\ref{sec:f1f2}, we have
\[
\{H_4^{\rm nr},F_1 \}^{\rm r}=2 i \sum_{\stackrel{m_4,m_5\neq m_1,\,\,m_2,m_3\neq m_6}{l(m)=0,\,\,q(m)=0}} \frac{v(m_1)v(m_2)v(m_3)\bar v(m_4)\bar v(m_5)\bar v(m_6)}
{(m_2-m_6)(m_3-m_6)} -c.c.
\]
Using duality as above we need to estimate 6 terms of the form
\begin{equation}\label{eq:h4nrr}
\sum_{\stackrel{m_4,m_5\neq m_1,\,\,m_2,m_3\neq m_6}{l(m)=0,\,\,q(m)=0}}\frac{v_1(m_1)v_2(m_2)v_3(m_3) v_4(m_4)v_5(m_5)v_6(m_6)}
{|m_2-m_6||m_3-m_6|},
\end{equation}
where $v_j=|h|$ and others are $|v|$. The required estimates follow from the following claim:
For any $\eta>0$, for any permutation $(j_1,j_2,j_3)$ of $\{1,4,5\}$, and for any permutation $(n_1,n_2,n_3)$ of $\{2,3,6\}$, we have
\begin{align*}
\eqref{eq:h4nrr}\les \|v_{j_1}\|_{\ell^\infty} \|v_{j_3}\|_{\ell^1}
\|v_{n_1}\|_{\ell^\infty} \|v_{n_3}\|_{\ell^1}\big(\|v_{j_2}\|_{\ell^\infty}\|v_{n_2}\|_{\ell^\infty}\big)^{\frac{1}{1+\eta}}
\big(\|v_{j_2}\|_{\ell^1}\|v_{n_2}\|_{\ell^1}\big)^{\frac{\eta}{1+\eta}}.
\end{align*}
As in the proof of Lemma~\ref{lem:F2}, see \eqref{eq:F2tilde}, the claim follows from an estimate for   the following sum
\begin{equation}\label{eq:h4nrr1}
\sum_{\stackrel{m_4,m_5\neq m_1,\,\,m_2,m_3\neq m_6}{l(m)=0,\,\,q(m)=0}}\frac{v_1(m_1)_2v(m_2)v_3(m_3) v_4(m_4)v_5(m_5)v_6(m_6)}
{|m_2-m_6|^{1+\eta}|m_3-m_6|^{1+\eta}}.
\end{equation}
First we replace $j_1$ in the equation $q(m)=0$ using $l(m)=0$. By symmetry it suffices to   consider two cases $j_1=1$, $j_1=4$. In the former case we have
\begin{align*}
0&=(m_2+m_3-j_2-j_3-m_6)^2+m_2^2+m_3^2-j_2^2-j_3^2-m_6^2\\
&=-2j_2(m_2+m_3-j_3-m_6)+(m_2+m_3-j_3-m_6)^2+m_2^2+m_3^2-j_3^2-m_6^2.
\end{align*}
Moreover, $m_2+m_3-j_3-m_6\neq 0$ since $m_1\neq m_4, m_5$. Therefore, both $j_1$ and $j_2$ are determined by the remaining indices. This implies that
\begin{align*}
\eqref{eq:h4nrr1}&\les \|v_{j_1}\|_{\ell^\infty} \|v_{j_2}\|_{\ell^\infty}\|v_{j_3}\|_{\ell^1}
\sum_{m_2,m_3\neq m_6}\frac{v(m_2)v_3(m_3)v_6(m_6)}
{|m_2-m_6|^{1+\eta}|m_3-m_6|^{1+\eta}}\\
&\les \|v_{j_1}\|_{\ell^\infty} \|v_{j_2}\|_{\ell^\infty}\|v_{j_3}\|_{\ell^1} \|v_{n_1}\|_{\ell^\infty} \|v_{n_2}\|_{\ell^\infty}\|v_{n_3}\|_{\ell^1},
\end{align*}
which leads to the desired estimate as in the previous sections.
The case $j_1=4$ is similar, the only difference is that $j_2$ is determined as roots of a quadratic polynomial instead of a linear one.

\subsection{Estimate of $ \D_{\bar v(k)}  \,  g_{F_2}^b g_{F_1}^a(H_4)$.}\label{sec:gf2h4}

The bounds for $\D_{\bar v(k)}  \,  g_{F_2}^b g_{F_1}^a(H_4)$ will be obtained inductively. Although we only need to consider the cases when  $a+b\geq 2$, we start with the case $a=1,b=0$ for clarity.
Note that $g_{F_1}^1(H_4)$ is a sum of terms of the form
$$H_4(v_1,v_2,v_3,v_4)=\sum_{n_1-n_2+n_3-n_4=0}v_1(n_1)  v_2(n_2)v_3(n_3) v_4(n_4)$$
where one of $v_i$'s is $f_1$ or $\bar f_1$ and the others are $v$ or $\bar v$.
To estimate $\big\|\D_{ \bar v(k)}   g_{F_1}^1(H_4)\big\|_{\ell^p}$, we use duality as before:
\ba\label{eq:duality}\Big\|\frac{\D}{\D \bar v(k)}   g_{F_1}^1(H_4)\Big\|_{\ell^p} \leq \sup_{\|h\|_{p^\prime}=1} \sum_k \Big|\frac{\D}{\D \bar v(k)}   g_{F_1}^1(H_4)\Big| |h(k)|.\ea
Note that the sum in the right hand side of \eqref{eq:duality} is bounded by  the sum of the following two terms  
\begin{align*}
 H_4(|f_1(|v|,|v|,|v|)|,|h|,|v|,|v|),\,\,\,\,\,\,\,\,\,\,\,
 H_4(|f_1(|h|,|v|,|v|)|,|v|,|v|,|v|)
\end{align*}
 and similar terms obtained by permuting the arguments.
The following lemma will be used to estimate these terms and the ones appearing in the higher order commutators. 

\begin{lem} \label{lem:induc} 
For any $q\in[1,\infty]$ and any permutation $(i_1,i_2,i_3,i_4)$ of $(1,2,3,4)$, we have
$$
|H_4(v_1,v_2,v_3,v_4)|\leq \|v_{i_1}\|_{\ell^q}\|v_{i_2}\|_{\ell^{q^\prime}}\|v_{i_3}\|_{\ell^1}\|v_{i_4}\|_{\ell^1}.$$
\end{lem}
\begin{proof} Note that for any permutation we can write
$$H_4(v_1,v_2,v_3,v_4) = \sum_j v_{i_1}(j)\,\,v_{i_2}*v_{i_3}*v_{i_4}(j).$$
The statement follows from H\"older's and Young's inequalities.
\end{proof}

Using  Lemma~\ref{lem:induc} and Lemma~\ref{lem:f1f2}, we obtain 
\begin{align*}
H_4(|f_1(|v|,|v|,|v|)|,|h|,|v|,|v|)&\les \|h\|_{\ell^{p^\prime}}\|v\|_{\ell^p}\|v\|_{\ell^1} \|f_1(|v|,|v|,|v|)\|_{\ell^1}\\
&\les \|h\|_{\ell^{p^\prime}}\|v\|_{\ell^p}\|v\|_{\ell^1} \|v\|_{\ell^1} \|v\|_{\ell^{\infty-}}^2 \\
& \les \|h\|_{\ell^{p^\prime}} \eps^{\frac12-\frac1p-}.
\end{align*}
Similarly, we have
\begin{align*}
H_4(|f_1(|h|,|v|,|v|)|,|v|,|v|,|v|)&\les \|f_1(|h|,|v|,|v|)\|_{\ell^{p^\prime}}\|v\|_{\ell^p}\|v\|_{\ell^1}^2 \\
&\les \|h\|_{\ell^{p^\prime}}\|v\|_{\ell^{\infty-}}^2 \|v\|_{\ell^p}\|v\|_{\ell^1}^2 \\
&\les \|h\|_{\ell^{p^\prime}} \eps^{\frac12-\frac1p-}.
\end{align*}
Similar bounds follow for the terms obtained by permuting the arguments. Therefore we have
$$\Big\|\frac{\D}{\D \bar v(k)}   g_{F_1}^1(H_4)\Big\|_{\ell^p}\les \eps^{\frac12-\frac1p-}.$$
Note that this gives an error of order 1 when $p=2$. This explains why we consider higher order commutators and a second normal form transform (see   footnote~\ref{fn:gf1h}).
This proof motivates the following generalization:
 \begin{lem}\label{lem:N}
 Consider $H_4(|v|,|v|,|v|,|v|)$. Repeatedly ($a$ times) replace one of the $v$'s with $f_1(|v|,|v|,|v|)$. Then  repeatedly ($b$ times) replace one of the $v$'s with $f_2(|v|,|v|,|v|,|v|,|v|)$. Finally, replace one of the $v$'s with $h$. We denote any such function by $H_{4,a,b}(f_1,f_2,h,v)$. Then, for  $p=1$ and $p=\infty$, we have
\begin{align*}
|H_{4,a,b}(f_1,f_2,h,v)|
&\les \|h\|_{\ell^{p^\prime}} \eps^{a+b-\frac12-\frac1p-}.
\end{align*}
\end{lem}
\begin{proof} First by using Lemma~\ref{lem:f1f2} repeatedly (with sufficiently small $\eta$, we see that any composition of $f_1$'s and $f_2$'s satisfy
\ba\label{eq:frep}
\|\cdot\|_{\ell^q}\les \|v\|_{\ell^q} \|v\|_{\ell^{\infty-}}^{2a}\Big[\|v\|_{\ell^1}^{1+ }\|v
\|_{\ell_\infty}^{3-}\Big]^b, 
\ea
where $a$ is the number of $f_1$'s and $b$ is the number of $f_2$'s appearing in the composition.

Now, note that $H_4$ has four arguments. Let $a_j$ (resp. $b_j$) be the number of $f_1$'s (resp. $f_2$'s) appearing in the $j$th argument. 
Only one of the arguments contains $h$, say the first one. Using Lemma~\ref{lem:induc}, we have
$$
|H_4(v_1,v_2,v_3,v_4)| \les \|v_{1}\|_{\ell^{p^\prime}}\|v_{2}\|_{\ell^p }\|v_{3}\|_{\ell^1}\|v_{4}\|_{\ell^1}.
$$
Using \eqref{eq:frep}, we have
\ba\nonumber
\|v_{2}\|_{\ell^p }\|v_{3}\|_{\ell^1}\|v_{4}\|_{\ell^1}\les \|v \|_{\ell^p }\|v \|_{\ell^1}\|v \|_{\ell^1}\|v\|_{\ell^{\infty-}}^{2(a_2+a_3+a_4)}\Big[\|v\|_{\ell^1}^{1+}\|v
\|_{\ell_\infty}^{3-}\Big]^{b_2+b_3+b_4}
\ea

Next, note that $v_1$ is either $|h|$ (in which case we stop) or $f_1(v_{1,1},v_{1,2},v_{1,3})$ or $f_2(v_{1,1},\ldots,v_{1,5})$. 
In the latter cases, without loss of generality, $v_{1,1}$  contains $|h|$. We estimate,   using \eqref{eq:frep} and a simple induction,
\begin{align}\nonumber
\|v_{1}\|_{\ell^{p^\prime}} \les \|h\|_{\ell^{p^\prime}}\|v \|_{\ell^{\infty-}}^{2a_1}\Big[\|v\|_{\ell^1}^{1+}\|v
\|_{\ell_\infty}^{3-}\Big]^{b_1}.
\end{align}
Combining these estimates we obtain
\begin{align*}
|H_{4,a,b}(f_1,f_2,h,v)| &\les \|h\|_{\ell^{p^\prime}} \|v\|_{\ell^p }\|v \|_{\ell^1}\|v \|_{\ell^1}
\|v \|_{\ell^{\infty-}}^{2a}\Big[\|v\|_{\ell^1}^{1+}\|v
\|_{\ell_\infty}^{3-}\Big]^{b}\\
&\les \|h\|_{\ell^{p^\prime}} \eps^{a+b-\frac12-\frac1p-}.
\end{align*}
\end{proof}
Using duality as above we see that the right hand side of \eqref{eq:duality} for $ \D_{ \bar v(k)}   g_{F_2}^b g_{F_1}^a(H_4)$ can be bounded by a finite sum of functions
$H_{4,a,b}(f_1,f_2,h,v)$. Therefore, Lemma~\ref{lem:N} implies that 
$$\Big\|\frac{\D g_{F_2}^b g_{F_1}^a(H_4)}{\D \bar v(k)}   \Big\|_{\ell^p} \les \eps^{a+b-\frac12-\frac1p-} \les \eps^{ \frac32-\frac1p-},\,\,\,\text{ if }a+b\geq 2.$$

\subsection{Remainder Estimates}
It remains to estimate the terms involving integrals.
Note that it suffices to prove the inequalities
\begin{align*}
&\sup_{s\in [0,1]}\Big\| \frac{\D}{\D \bar v(k)} g_{F_1}^3(H)\circ X_{F_1}^s
\Big\|_{\ell^p}\les \eps^{\frac32-\frac1p-},\\
&
\sup_{s\in [0,1]}\Big\| \frac{\D}{\D \bar v(k)} \{g_{F_1}^3(H)\circ X_{F_1}^s,F_2\}\Big\|_{\ell^p}\les \eps^{\frac32-\frac1p-},\\
&\sup_{s\in [0,1]}\Big\|\frac{\D}{\D \bar v(k)} 
g_{F_2}^{2}(H\circ X_{F_1}^1)\circ X_{F_2}^s\Big\|_{\ell^p}\les \eps^{\frac32-\frac1p-}
\end{align*}
for $p=1,\infty$ assuming that $\|v\|_{\ell^p}\les \eps^{\frac12-\frac1p}, p\in[1,\infty]$. 
Since we have to estimate the composite function derivative, we first
study the bounds on the derivatives of $X_{F_j}^s(v)$, $j=1, 2$, $s\in[0,1]$, more precisely, let 
$
w(m)=[X_{F_j}^s(v)](m),
$
which is the solution at $t=s$ of the system
\[
\frac{d w(m)}{d t}=\frac{\partial F_j}{\partial \bar w(m)}, \,\,\,\,\,\,w|_{t=0}= v.
\]
Differentiating this equation with respect to initial condition
$ w(n)|_{t=0}= v(n)$ and using the notation $D_n$, we see that $\big|\frac{d}{dt}D_n w(m) \big|$ is bounded by a sum of terms of the form
\begin{align*}&f_1(v_1,v_2,v_3)(m),\,\,\,\,\text{ for } j=1, \\
&f_2(v_1,v_2,v_3,v_4,v_5)(m),\,\,\,\,\text{ for } j=2,
\end{align*}
where one of the $v_k$'s is $|D_n w|$ and the others are $|w|$. Without loss of generality we can assume that 
$v_1=|D_nw|$. 
We have a similar formula for $\frac{d}{ds}D_n \bar w$.
Note that at $s=0$, we have
$$\Big\| D_nw(m) \Big\|_{\ell^\infty_m\ell^1_n}=\Big\| D_nw(m) \Big\|_{\ell^\infty_n\ell^1_m } =1.$$
We will prove that both of these norms remain bounded for $s\in[0,1]$. 
Taking the $\ell^\infty_m\ell^1_n$ norm of $f_j$ 
we obtain 
\begin{align*}  
\Big\|\frac{d}{dt}   D_n w(m)\Big\|_{\ell^\infty_m\ell^1_n} & \les \big \|f_j(|D_n w|,\ldots,|w|)(m)\big\|_{\ell^\infty_m\ell^1_n}  \\ 
&\leq \big \|f_j(\|D_n w\|_{\ell^1_n},\ldots,|w|)(m)\big\|_{\ell^\infty_m}  
\\
& \les \|D_n w(m)\|_{\ell^\infty_m\ell^1_n} \eps^{1-}.
\end{align*}
In the last line, we used Lemma~\ref{lem:f1f2} (for sufficiently small $\eta$).
This implies that (with $w(m)=[X_{F_j}^s(v)](m)$, $j=1,2$)
\ba\label{eq:deriv}
\sup_{0\leq s\leq 1} \Big\| D_n w(m)\Big\|_{\ell^\infty_m\ell^1_n} \les 1.
\ea
Similarly, we obtain
\ba\label{eq:deriv1}
\sup_{0\leq s\leq 1} \Big\| D_n w(m)\Big\|_{\ell^\infty_n\ell^1_m} \les 1.
\ea
We also need the following estimates for the higher order derivatives of $w=X^s_{F_1}(v)$ with respect to the initial conditions:
\begin{align*}
\|D_jD_nw(k)\|_{\ell^\infty_{j,n}\ell^1_k }&\les \eps^{\frac12-},\,\,\,\,\,\,\,
\|D_jD_nw(k)\|_{\ell^\infty_{k,n}\ell^1_j }\les \eps^{\frac12-}, \\
\|D_jD_mD_nw(k)\|_{\ell^\infty_{j,m,n}\ell^1_k }&\les 1,\,\,\,\,\,\,\,
\|D_jD_mD_nw(k)\|_{\ell^\infty_{k,m,n}\ell^1_j }\les 1,
\end{align*}
which can be obtained using Lemma~\ref{lem:f1f2} as in the proof of \eqref{eq:deriv}, \eqref{eq:deriv1}.

\begin{rmk}
For $\delta>0$, a similar argument implies  
\[
\Big\|e^{\delta|n-m|} D_n w(m)\Big\|_{\ell^\infty_m\ell^1_n}\les 1,\,\,\,\,\,\,\,\, 
\Big\|e^{\delta|n-m|} D_n w(m) \Big\|_{\ell^\infty_n\ell^1_m}  \les 1,
\]
and for higher order derivatives of $w=X^s_{F_1}(v)$ we have
\begin{align*}
\Big\|e^{\delta|j_1+\cdots +j_k-m|} D_{j_1}\ldots D_{j_k} w(m)\Big\|_{\ell^\infty_{j_1,\ldots,j_k}\ell^1_m}&\les 1,\\
\Big\|e^{\delta|j_1+\cdots +j_k-m|} D_{j_1}\ldots D_{j_k} w(m) \Big\|_{\ell^\infty_{m,j_2,\ldots,j_k} \ell^1_{j_1}}  &\les 1.
\end{align*}
The rest of the argument follows as in other sections.
\end{rmk}

\subsubsection{Estimation of $\D_{\bar v(k)}g_{F_1}^3(H)\circ X^s_{F_1}(v) $.} \label{sec:subsub}
Since $g_{F_1}(\Lambda_2) = -H_4^{\rm nr}$, it suffices to estimate 
\[
\sup_{s\in [0,1]}\Big\| \frac{\D}{\D \bar v(k)} g_{F_1}^a(H_4)\circ X^s_{F_1}(v)
\Big\|_{\ell^p}, \,\,\,\,a=2,3.\]
When $a=2$, we estimate this expression rather than the one containing $H_4^{\rm nr}$
(as we should have) because it simplifies the notation and still implies
the estimate for the required expression. Note that
\begin{align*} 
 \Big\| \frac{\D}{\D \bar v(k)} g_{F_1}^a(H_4)\circ X^s_{F_1}(v) \Big\|_{\ell^1_k} &\leq  \Big\| \sum_i
\frac{\D g_{F_1}^a (H_4)}{\D \bar w(i)} \frac{\D \bar w(i)}{\D \bar v(k)}\Big\|_{\ell^1_k} + \|c.c.\|_{\ell^1_k}\\
&\leq \Big\|\frac{\D g_{F_1}^a (H_4)}{\D \bar w(i)}\Big\|_{\ell^1_i} \Big\| D_k \bar w(i)\Big\|_{\ell^\infty_i\ell^1_k}+ \|c.c.\|_{\ell^1_k} \\
&\les \eps^{\frac12-},
\end{align*} 
 for   $a\geq 2$ by \eqref{eq:deriv} and the estimates we obtained in Subsection~\ref{sec:gf2h4}. Similarly, using \eqref{eq:deriv1}, we obtain
\begin{align*} 
\Big\|\frac{\D}{\D \bar v(k)} g_{F_1}^a(H_4)\circ X_{F_1}^s(v) \Big\|_{\ell^\infty_k}\les \eps^{\frac32-},\,\,\,\text{ for } a\geq 2.
\end{align*} 

The estimates for $\D_{\bar v(k)}\{g_{F_1}^3(H)\circ X_{F_1}^s,F_2\}  $ and 
$\D_{\bar v(k)}g_{F_2}^{2}(H\circ X_{F_1}^1)\circ X_{F_2}^s $ are similar. The only difference is that we also require the higher order derivative estimates of $w=X_{F_1}^s(v)$ listed above.
We omit the details.

\appendix

\section{Nonlinear fiber optics application}
One of the most important applications of NLS 
concerns light-wave communication systems, where optical pulses in a retarded time frame evolve 
according to the one dimensional NLS
\ba
iA_z + S d(z) A_{\tau \tau} + g(z) |A|^2 A = 0.
\ea
Here $z$ is the rescaled distance, $\tau$ is the rescaled retarded time, $A$ is the amplitude of 
the optical wave envelope,
$d(z)$ is the group velocity dispersion, which is usually piecewise constant, and $S$ is the dispersion strength 
parameter. Finally, $g(z)>0$ is 
the nonlinear coefficient which accounts for the losses and amplifications. 
 For the derivation of NLS from Maxwell's equations, one can consult 
many references, {\em e.g.} \cite{HasKod}.  It is a standard assumption 
that $d(z)$ and $g(z)$ are periodic.

In general, in light-wave communication,  the information is transmitted with localized pulses 
(with Gaussian or exponential tails) in allocated time slots. 
The presence of pulse corresponds to ``1'' and the absence of pulse corresponds to ``0''
in binary format.
Naturally, it is preferable that the incoming waveform would appear undistorted at the end of 
the transmission line. It can be achieved  by optimizing  
an individual pulse, so it would propagate without distortion, and sending such pulses together, keeping 
them sufficiently far apart ({\em i.e.} taking time slots sufficiently large), 
so they would not interact. 
Such regime is usually called ``soliton regime'' in the optical communication 
literature, where the word  ``soliton'' does not usually mean that 
the equation   is integrable. The pulses could be, for example, dispersion managed solitons, 
which are  approximately periodic localized solutions of the above equation. 
In other words, the main feature of the 
soliton regime is that the pulses do not interact (or rather pulse 
to pulse interaction is weak compared 
to the pulse self-interaction) during the propagation through the transmission line.

An alternative regime (often called the quasi-linear regime) has been found when 
the pulses strongly overlap during the transmission, 
see {\em e.g.} the survey paper \cite{Ess}. Surprisingly, it was observed that up to a   
linear transformation of the transmitted waveform, the pulses appeared undistorted.   Note that even though the
pulses spread over many time slots, the average optical energy ($L^2$ norm square) per bit 
does not change and therefore nonlinear effects remain strong. 
It is usually implicitly assumed  in the engineering literature that 
``nonlinearity gets averaged out'' due to the high frequency of the initial data. 

In this article, we rigorously explain the quasi-linear phenomenon  for a model problem  
when $d(z)$ and $g(z)$ are constant and all bits are occupied by $1's$, in 
the limit  of vanishing pulse width. 
This case (of all identical 1's) leads to the formulation with periodic boundary conditions. 
Although, this is a special case, we hope that our proof  can be extended to  the more general case: 
pseudo-random sequence of $1's$ and $0's$. 
Note that constant $d(z)$ and $g(z)$ assumption is not restrictive since if the evolution is quasi-linear on each interval where $d(z),g(z)$ are constant, then the evolution is quasi-linear on their union.

There has been previous work on quasi-linear regime. In \cite{ManZak}, 
the limit of the short pulse width  for dispersion managed NLS  on the real line is considered. 
An effective 
evolution equation was derived which turned out to be integrable and weakly nonlinear. The equation  was later improved 
in \cite{Abl}. 
On the real line the energy disperses to infinity and therefore nonlinearity becomes small. 
This leaves an open question: {\em what will happen if the energy does not disperse 
to infinity or in other words, there is an  infinite bit stream}. The problem considered in this paper models
this situation: nonlinearity remains strong which is due to the periodic boundary conditions.

Finally, we note that on $\Real$, the dispersion strength $S$ and  pulse width $\eps$ can be combined into 
a single effective  parameter $S/\eps^2$ by scaling $\tau$. This implies that the limits $S\rightarrow \infty$ 
and $\eps\rightarrow 0$ are equivalent. This is not the case in our model since
the characteristic $\tau$-scale, bit size,  is already present.
Therefore, the two parameter problem in $S, \eps$ 
should be considered. However, the limit $S\rightarrow \infty$ is insufficient to achieve 
quasi-linear evolution and must be supplemented with $\eps \rightarrow 0$. On the other hand, 
the limit $\eps \rightarrow 0$ does produce quasi-linear evolution with $S$ being fixed but arbitrary. 
This motivated us to consider only this case. We put  $S=1$ in order not to obscure the exposition. \\ \\
\begin{large}
{\bf Acknowledgment. \\}
\end{large}
The authors would like to thank Eugene Wayne for many helpful 
discussions.

\end{document}